\documentclass[a4paper, 10pt, conference]{IEEEtran}

\usepackage{amsmath,amssymb,graphicx}
\usepackage[noadjust]{cite}
\usepackage{filecontents}
\usepackage{amsthm}
{
	\theoremstyle{plain}
	\newtheorem{assumption}{Assumption}
}
\usepackage{xcolor}
\usepackage{ulem}
\usepackage{hyperref}
\usepackage{lipsum}
\usepackage{epstopdf}

\usepackage{caption}
\usepackage{subcaption}
\IEEEeqnarraydefcol{myp}{\parbox[c]{0.5in}}{}
\IEEEoverridecommandlockouts
\newtheorem{theorem}{Theorem}

\usepackage{wrapfig}

\providecommand{\keywords}[1]{\textbf{\textit{Index terms---}} #1}

\begin{document}
	\title{Multivariable Newton-Based Extremum Seeking with Bounded Update Rates}
\author{
	\IEEEauthorblockN{\fontsize{12pt}{13pt}\selectfont
		Farzaneh Karimi\IEEEauthorrefmark{1},
		Mohsen Mojiri\IEEEauthorrefmark{1},
		Mohammadali Ghadiri-Modarres\IEEEauthorrefmark{2},
		Azadeh Mohammadpour\IEEEauthorrefmark{3}
	}
	\IEEEauthorblockA{\small\IEEEauthorrefmark{1}
		Department of Electrical and Computer Engineering,
		Isfahan University of Technology, Isfahan, Iran
	}
	\IEEEauthorblockA{\small \IEEEauthorrefmark{2}
		Department of Electrical Engineering,
		Arak University of Technology, Arak, Iran
	}
	\IEEEauthorblockA{\small\IEEEauthorrefmark{3}
		Department of Electrical Engineering,
		Shahid Beheshti University, Tehran, Iran
	}
}
\newtheorem{thm}{Theorem}
\newtheorem{prf}[thm]{Proof} 
\newtheorem{definition}{Definition}
\twocolumn[
\begin{@twocolumnfalse}
	\maketitle
	\vspace{2mm}
	\noindent\rule{\textwidth}{0.5pt}
	\vspace{2mm}
\begin{abstract}
	We propose a multivariable Newton-based extremum seeking scheme, incorporating known bounds on update rates. The design extends the recent bounded extremum seeking scheme to the Newton-based approach, which makes the convergence rate independent of the unknown Hessian by estimating and inverting the Hessian matrix. As a vital part of the Newton-based algorithm, we design an appropriate demodulation matrix to generate an estimate of the Hessian in an average sense. The local exponential stability of the bounded extremum seeking algorithm is proven for general multivariable static maps using averaging analysis. In comparison with previous practical stability results, it guarantees convergence to a neighborhood of the extremum point with an exact convergence rate exhibiting exponential decay. Subsequently, the main stability result is established for the Newton-based bounded extremum seeking approach under a weak Hessian coupling assumption. Simulation results demonstrate the advantage of the proposed approach over bounded extremum seeking by assigning equal, desired convergence rates to all parameters using the Newton approach.
\end{abstract}
	\vspace{2mm}
	\noindent\rule{\textwidth}{0.5pt}
	\vspace{2mm}
	\end{@twocolumnfalse}
]
	\section{Introduction}\label{sec1}
	The classical extremum-seeking (ES) approach is gradient-based and typically employs sinusoidal perturbation signals to estimate the gradient of an unknown static map. Despite its broad theoretical development and practical applications  \cite{wang2025distributed, scheinker2021extremum, tabatabaei2025voltage}, the convergence rate of gradient-based ES schemes generally depends on the local curvature of the map, which, in the multivariable setting, is characterized by its Hessian matrix. Since this curvature is usually unknown, the resulting convergence rate cannot generally be prescribed a priori. Moreover, the parameter-update rates are not inherently bounded in conventional gradient-based ES schemes.
	
	Bounded update rates are important in practical applications. For example, in source-localization problems, physical and actuator constraints may impose explicit limits on the rate at which the control parameters can change \cite{ghadiri2016new}. To address this issue, \cite{scheinker_extremum_2014} introduced ES with bounded update rate (Bounded ES) in which the argument of the sinusoidal demodulation functions is used to confine the unknown map value, thereby guaranteeing bounded parameter-update rates. The stability of this approach was established using the results in \cite{kurzweil1987limit,sussmann1992new}, yielding global practical stability. More recently, \cite{zhu2022extremum} employed a time-delay approach to averaging to obtain practical-stability results and explicit bounds on the seeking error. However, these results are restricted to one- and two-variable problems and require restrictive prior information, including bounds on the extremum point, the extremum value, and the Hessian matrix.
	
	Although Bounded ES schemes constrain the parameter-update rates, their convergence rates remain sensitive to the unknown curvature of the map. To reduce this dependence, Newton-based ES schemes have been developed, including the approach proposed in \cite{moase_newton-like_2010}. By incorporating an estimate of the inverse Hessian, Newton-based ES can compensate for the effect of the local curvature on the convergence rate. However, combining this convergence-rate advantage with explicitly bounded parameter-update rates remains challenging. In particular, a suitable Hessian-estimation mechanism and a corresponding stability analysis are required for a Newton-based bounded ES scheme.

	The main contributions of this paper are as follows. First, we establish local exponential stability for the averaged Bounded ES system using averaging analysis. Compared with the practical-stability results in \cite{scheinker_extremum_2014}, the proposed analysis provides an explicit exponential decay estimate and a quantitative convergence rate for the averaged error dynamics, as well as the corresponding neighborhood to which the original oscillatory system converges.  Second, we design a Hessian estimator based on an appropriately constructed demodulation matrix and incorporate its inverse into the bounded ES update law. Third, we establish local exponential stability of the resulting Newton-based bounded ES scheme under a weak Hessian-coupling condition. The proposed approach allows equal desired convergence rates to be assigned to the optimization variables while respecting prescribed bounds on their update rates.
	
	The remainder of this paper is organized as follows. Section \ref{ProblemStatement} introduces the gradient-based bounded ES scheme and establishes its local exponential convergence properties. Section \ref{newton} presents the proposed Newton-based bounded ES scheme and provides its stability analysis. Section \ref{Sim} illustrates the theoretical results through a simulation example. Finally, Section \ref{CoN} summarizes the main conclusions.
\subsection{Notation and definitions}
Throughout the article, vectors and matrices are denoted in boldface. For clarity and conciseness, we use the following compact notation:
\begin{enumerate}
	\item $\sin(\boldsymbol{\scalebox{0.85}{${\Omega}$}}t+\boldsymbol{\phi})=[\sin({\omega_1}t+{\phi_1}),\sin({\omega_2}t+{\phi_2}),...,\sin({\omega_n}t+{\phi_n})]^T$. Where $\boldsymbol{\scalebox{0.85}{${\Omega}$}}$ is a frequency vector with elements $\omega_i$ and $\boldsymbol{\phi}$ is defined as a vector containing the phase shifts $\phi_i$. The same is true for the cosine function.
	\item If $\boldsymbol{\mu}$ is a vector with elements $\mu_i$, the diagonal matrix corresponding to the vector $\boldsymbol{\mu}$ is defined as $\operatorname{diag} (\boldsymbol{\mu})=\operatorname{diag} (\mu_1, \mu_2,..., \mu_n)$.
\end{enumerate}
Also Table 1 shows the notations adopted in this paper.
\begin{table}[h!]
	\begin{center}
		\caption{Notation }
		\label{tab:table1}
		\begin{tabular}{c|l} 
			\hline
			\textbf{Notation} &\textbf{Definition}\\
			\hline
			\hline \\
			
			\shortstack{$\mathbb{R}$\\ \\ \hspace{1cm}}&\shortstack[l]{Set of real numbers\vspace{0.25cm}} \\ 
			
			\shortstack{$\mathbb{R}^n$\\ \\ \hspace{1cm}} &\shortstack[l]{Set of n-dimensional real vectors\vspace{0.25cm}} \\
			
			\shortstack{$x^T$\\ \\ \hspace{1cm}}&\shortstack{The transpose of the vector $x\in\mathbb{R}^n$\vspace{0.25cm}}\\
			\shortstack{$	\vert|.	\vert| $ \\ \\ \hspace{1cm}} & \shortstack{ Standard Euclidean norm\vspace{0.25cm}} \\ 
			
			\shortstack{	$\nabla V$ \\ \\ \hspace{1cm}}& \shortstack{ Gradient of real-valued differentiable function $V$ \vspace{0.25cm}}\\
			\shortstack{	$f_{av}(x)$ \\ \\ \\ \\ \\ \\ \\ \\ \hspace{1cm}}& \shortstack[l]{ Well-defined average of periodic function $f(x,t)$, defined \\as $f_{av}(x) = \frac{1}{T} \int_{0}^{T} f(x,t) \, dt$, where $T$ is the period time \vspace{0.25cm}}
		\end{tabular}
	\end{center}
\end{table}
\section{Extremum Seeking with bounded update rate }\label{ProblemStatement}
\subsection{Introduction to the Existing Method \cite{scheinker_extremum_2014}}
ES scheme in multivariable case considers applications
in which the goal is to maximize (or minimize) the scalar
output $y \in \mathbb{R}$ of an unknown and convex nonlinear static map
$y = J(\boldsymbol{\theta})$ by varying the input vector $\boldsymbol{\theta}=\left[\theta_1,...,\theta_n\right]^T\in \mathbb{R}^n$. 
In maximum/ minimum seeking problem, there exists  $\boldsymbol{\theta^*}\in \mathbb{R}^n$ such that
\begin{align}
	\nabla J(\boldsymbol{\theta})\Big|_{\boldsymbol{\theta^*}}\!\!\!\!= 0
	\label{Jcondition}
\end{align}
By introducing the bounded ES in \cite{scheinker_extremum_2014}, the classical multivariable ES was improved such that the update rate is known and bounded. It is achieved by the argument of a sine/cosine function confines the uncertainty.

In the multivariable bounded ES method, the case of minimization of a measurable output, the update law is defined as
\begin{align}
	&\dot{\theta}_i=\sqrt{\alpha_i\omega_i}\cos(\omega_it+\mu_iJ)\hspace{1cm} i=1,2,...,n
	\label{update_law_ith_shanker}
\end{align} 
where $\omega_i$ is the probing frequency and should be chosen  \( {\omega}_i \neq{\omega}_j \) for all \( i \neq j \). Also, $\alpha_i$ is nonzero and small amplitude.
\begin{theorem}\cite{scheinker_extremum_2014}
	Consider the multivariable bounded ES scheme with the update law \eqref{update_law_ith_shanker}
	and the map \( J\) satisfies \eqref{Jcondition}. The minimum value \( \boldsymbol{\theta^*}\in \mathbb{R}^n \) is \( \frac{1}{\omega} \)-SPUAS\footnote{$\frac{1}{\omega}$-Semiglobal Practical Uniform Asymptotic Stability}.
\end{theorem}
The proof of this theorem is provided in \cite{scheinker_extremum_2014} by combining the Lie bracket-based averaging results from \cite{durr2013lie} with the Lyapunov function analysis techniques from \cite{moreau2002practical}.
The converging trajectories property is an essential aspect of SPUAS as it ensures that trajectories of the system converge to a certain bounded region. For a deeper insight into this property and its implications, refer to  \cite{moreau2002practical}.
	\subsection{Local Exponential Stability Analysis}\label{stabilityGrad}
	In this section, we build upon the findings of \cite{scheinker_extremum_2014} to demonstrate local exponential stability of update law \eqref{update_law_ith_shanker}. To keep things straightforward, the compact form (vector form) for multivariable case is considered. In maximum seeking problem, there exist  $\boldsymbol{\theta^*}\in \mathbb{R}^n$ such that, \eqref{Jcondition} is satisfied for map $J(\boldsymbol{\theta})$. 

	The nonlinear map around the peak $\boldsymbol{\theta^*} $ can be rewrite by the second order Taylor series expansion 
	\begin{align}
		J(\boldsymbol{\theta})=J(\boldsymbol{\theta^*})+\frac{1}{2}(\boldsymbol{\theta}-\boldsymbol{\theta^*})^T \boldsymbol{H }(\boldsymbol{\theta}-\boldsymbol{\theta^*})
		\label{J}
	\end{align}
	where $\boldsymbol{H}$ is the unknown Hessian of the map and is assumed to be negative definite for the maximum-seeking problem.
	
	The multivariable bounded ES scheme in compact form is shown in Fig. \ref{figINT}.	The update law \eqref{update_law_ith_shanker} is governed by
	\begin{align}
		&\boldsymbol{\dot{\theta}}=\boldsymbol{{K_1}}\sin(\boldsymbol{\scalebox{0.85}{${\Omega}$}} t-\boldsymbol{\mu} J)
		\label{update_law_vector}
	\end{align}
	where $\boldsymbol{{K_1}}=\operatorname{diag} \left(\sqrt{\alpha_1 \omega_1},..., \sqrt{\alpha_n \omega_n}\right)$, $\boldsymbol{\mu}=\left[\mu_1, ..., \mu_n\right]^T$  and $\boldsymbol{\scalebox{0.85}{${\Omega}$}}=\left[\omega_1, ..., \omega_n\right]^T$ such that $\alpha_i>0$ and $\mu_i>0$ are tuning parameters and $\omega_i$'s are the probing frequencies which can be selected as
	\begin{align}
		\omega_i=\overline{\omega}_i\omega=O(\omega), \hspace{1cm}i\in {1,2,...,n}
	\end{align}
	where $\omega$ is a positive constant and $\overline{\omega}_i$ is a rational number. One possible choice is given in \cite{ghaffari_multivariable_2012} as
	\begin{align}
		\overline{\omega}_i\notin \{\overline{\omega}_j, \frac{1}{2}(\overline{\omega}_j+\overline{\omega}_k), \overline{\omega}_j+2\overline{\omega}_k, \overline{\omega}_j+\overline{\omega}_k\pm\overline{\omega}_l\}
		\label{W_isump}
	\end{align}
	for all distinct $i,j,k,l\in\{1,...,n\}$. Also, in multivariable case,
	it is necessary to choose a sufficiently large ${\omega}$ as well as
	distinct probing frequencies (${\omega_i}\neq {\omega_j}$), and ensure that
	the ratio ${\omega_i}/{\omega_j}$ is rational and ${\omega_i}+{\omega_j}\neq {\omega_k}$ for different
	$i, j, k$. 
	The stability properties of the system \eqref{update_law_vector} are stated in the following theorem.
	\begin{theorem}
		Consider the ES scheme shown in Fig.~\ref{figINT} with the update law \eqref{update_law_vector}. 
		There exist positive constants $\omega^*$ and $\alpha^*$ such that for all perturbation frequencies $\omega>\omega^*$ and \(0\!<\!\alpha_i\!<\!{\alpha}^*\), \(i\!=\!1,\ldots,n\),
		the solution $\boldsymbol{\theta}(t)$ of \eqref{update_law_vector} locally exponentially 
		converges to a periodic function within  $O\!\left(\frac{1}{\omega}+\underset{i}{\max}\sqrt{\frac{\alpha_i}{\omega_i}}\right)$-neighborhood 
		of $\boldsymbol{\theta}^*$.
		\label{shanker_static_m}
	\end{theorem}
	
	\begin{figure}[t]
		\centering
		\includegraphics[scale = .75]{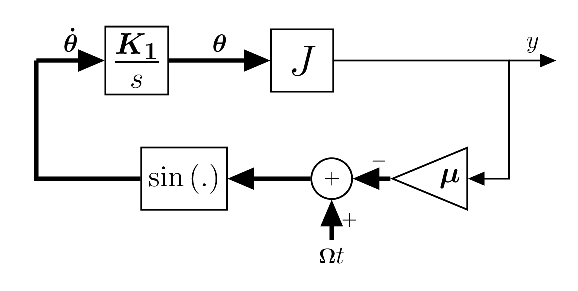} 
		\caption{ Multivariable Bounded ES  scheme.}
		\label{figINT}
	\end{figure}
	\begin{proof}
		Consider the change of variables as follows
		\begin{equation}
			\boldsymbol{\tilde{\theta}}=\boldsymbol{\theta}-\boldsymbol{\theta^*}+
			\boldsymbol{K_2}\cos(\boldsymbol{\scalebox{0.85}{${\Omega}$}} t-\boldsymbol{\mu}J)
			\label{tracking-error}
		\end{equation}
		where $
		\boldsymbol{	{K_2}}=\operatorname{diag}\left({\sqrt{\frac{\alpha_1}{\omega_1}}, ...,\sqrt{\frac{\alpha_n}{\omega_n}}}\right)$. Its time-derivative can be expressed as 
		
		\begin{equation}
			\boldsymbol{\dot{\tilde{\theta}}}=\boldsymbol{\dot{\theta}}-
			\boldsymbol{{K_2}}\left(
			\operatorname{diag}(\boldsymbol{\scalebox{0.85}{${\Omega}$}})-
			\operatorname{diag}(\boldsymbol{\mu})\frac{dJ}{dt}\right)\sin(\boldsymbol{\scalebox{0.85}{${\Omega}$}} t-\boldsymbol{\mu}J)
			\label{tracking-error_D}
		\end{equation}
		and by placing the update law (\ref{update_law_vector})
		\begin{equation}
			\boldsymbol{\dot{\tilde{\theta}}}=
			\boldsymbol{	{K_2} }
			\operatorname{diag}(\boldsymbol{\mu})\frac{dJ}{dt}\sin(\boldsymbol{\scalebox{0.85}{${\Omega}$}} t-\boldsymbol{\mu}J)
			\label{tracking-error_D1}
		\end{equation}
		$\frac{dJ}{dt}$ with respect to \eqref{update_law_vector} and \eqref{tracking-error}  can be taken as
		\begin{align}
		\nonumber	\frac{dJ}{dt}&=(\boldsymbol{\theta}-\boldsymbol{\theta^*})^TH\boldsymbol{\dot{\theta}}=(\boldsymbol{\tilde{\theta}}-\boldsymbol{K_2}\cos(\boldsymbol{\scalebox{0.85}{${\Omega}$}} t-\boldsymbol{\mu}J))^T\boldsymbol{HK_1}\\
		&	\hspace{4.3cm}\times\sin(\boldsymbol{\scalebox{0.85}{${\Omega}$}} t-\boldsymbol{\mu} J)
		\end{align}
		so we have 
		\begin{align}
			\nonumber			\boldsymbol{\dot{\tilde{\theta}}}=
			\boldsymbol{	{K_2} }
			\operatorname{diag}(\boldsymbol{\mu})&\sin(\boldsymbol{\scalebox{0.85}{${\Omega}$}}t-\boldsymbol{\mu}J)\sin^T(\boldsymbol{\scalebox{0.85}{${\Omega}$}} t-\boldsymbol{\mu}J)\\
			&\times\boldsymbol{K_1H}(\boldsymbol{\tilde{\theta}}-\boldsymbol{K_2}\cos(\boldsymbol{\scalebox{0.85}{${\Omega}$}} t-\boldsymbol{\mu}J))
			\label{tracking-error_D12}
		\end{align}
		In the new time-scale $\tau=\omega t$, we can rewrite \eqref{tracking-error_D12} as  
		\begin{align}
			\nonumber\frac{d\boldsymbol{{\tilde{\theta}}}}{d\tau}=
			\epsilon \boldsymbol{K_2} \operatorname{diag}({\boldsymbol{\mu}})&\sin(\overline{\boldsymbol{\scalebox{0.85}{${\Omega}$}}}\tau-{\boldsymbol{\mu}}J)\sin^T(\overline{\boldsymbol{\scalebox{0.85}{${\Omega}$}}}\tau-{\boldsymbol{\mu}} J){\boldsymbol{K_1}}
			\\& \times \boldsymbol{H}\left(\boldsymbol{\tilde{\theta}}-\boldsymbol{K_2}\cos(\overline{\boldsymbol{\scalebox{0.85}{${\Omega}$}}}\tau-{\boldsymbol{\mu}}J)\right)
			\label{tracking-error_D_nn}
		\end{align}
		where $\epsilon=\frac{1}{\omega}$ and it is small enough for $\omega$ to be sufficiently large.
		The system \eqref{tracking-error_D_nn} is $\Pi$-periodic in $\tau$ and in the standard form of the averaging theorem, where
		\begin{align}
			\Pi=2\pi\times LCM\{\frac{1}{\overline{\omega}_1},...,\frac{1}{\overline{\omega}_n}\}
		\end{align}
		$LCM$ stands for the least common multiple. We invoke the averaging theorem \cite[Th.~10.4]{khalil2002nonlinear} to analyze the stability properties. 
		Using the definition of averaging operator, the average system is
		\begin{equation}
			\boldsymbol{\dot{\tilde{\theta}}_{av}}=\epsilon\boldsymbol{KH}\boldsymbol{\tilde{\theta}_{av}}
			\label{CloseFormTav}
		\end{equation}
		where $\boldsymbol{{K}}=\operatorname{diag} \left(\frac{\mu_1\alpha_1}{2},..., \frac{\mu_n\alpha_n}{2} \right)$.
		\( \boldsymbol{\tilde{\theta}_{av}}= 0 \) is a locally exponentially stable equilibrium point for average system \eqref{CloseFormTav}. 
		Therefore, the application of the averaging theorem  proves the existence of a unique exponentially 
		stable periodic solution  \( {\boldsymbol{\tilde{\theta}}} \) of period
		\( \Pi \) within $O\!\left( \frac{1}{\omega} \right)$-neighborhood of equilibrium point system \eqref{CloseFormTav}, so we have
		\begin{align}
			\|{\boldsymbol{\tilde\theta}}\|\!=\!O\left( \frac{1}{\omega} \right) \!\Rightarrow\! \|\boldsymbol{\theta}-\boldsymbol{\theta^*}+\boldsymbol{	{K_2}}\cos(\boldsymbol{\scalebox{0.85}{${\Omega}$}} t-\boldsymbol{\mu}J)\|\!=\!O\left( \frac{1}{\omega} \right)
		\end{align}	
		using the triangle inequality, we obtain
		\begin{align}
			\|\boldsymbol{\theta} - \boldsymbol{\theta}^*\| \leq O\!\left(  \underset{i}{\max}\sqrt{\frac{\alpha_i}{\omega_i}}+ \frac{1}{\omega} \right) 
		\end{align}
		and, the proof is completed.
	\end{proof}
	From the averaged dynamics \eqref{CloseFormTav}, it can be observed that the convergence rate is governed by the unknown Hessian matrix $\boldsymbol{H}$.
	It is worth mentioning that the exact convergence rate cannot be determined because the ES algorithm is model-free, and we lack information about the Hessian matrix. This challenge is even more significant in multivariable systems, where additional parameters require tuning. In the next section, the Newton-based method, which has been introduced to address this problem, will be explored to establish for the bounded ES.
	
	\section{multivariable Newton-Based Extremum Seeking with bounded update rate}\label{newton}
	The gradient-based ES aims to make the gradient of the map approach zero, which means the system is being guided toward the optimal or extremum point. This principle also applies to the Newton-based ES, but the Newton approach goes a step further: it compensates for the effect of the map’s curvature (the shape of the cost surface) on how fast convergence happens. As a result, the Newton-based ES can often reach the extremum more quickly and accurately. 
	To achieve this, however, the algorithm must estimate not only the gradient but also the Hessian matrix of the underlying map, which represents how the gradient changes in different directions.
	\subsection{Proposed Approach}
	The proposed structure of the Newton-based ES scheme with a bounded update rate for a multivariable static map is illustrated in Fig \ref{fig2}. As mentioned in \cite{moase_newton-like_2010}, there are two essential parts to Newton-based ES scheme:
	\begin{enumerate}
		\item Estimate of the Hessian (shown as ${\boldsymbol{\hat{H}}}$ in Fig \ref{fig2}) using a suitable demodulation matrix (shown as ${\boldsymbol{N(t)}}$).
		\item Estimate and apply the inverse Hessian to the adaptive loop using a Ricatti filter.
	\end{enumerate}
	We summarize the system in Fig \ref{fig2} as  follows:
	\begin{subequations}
		\begin{align}
			\label{T1}
			&\boldsymbol{\dot{\theta}}=\boldsymbol{K_1}\sin \left(\boldsymbol{\scalebox{0.85}{${\Omega}$}} t+{\boldsymbol{\Gamma}}\boldsymbol{\mu}J\right)\\
			\label{T2}
			&\dot{{\boldsymbol{\Gamma}}}=\omega_r{\boldsymbol{\Gamma}}-\omega_r{\boldsymbol{\Gamma}} {\boldsymbol{\hat{H}}}{\boldsymbol{\Gamma}}
		\end{align}
		\label{System2}
	\end{subequations}
	where $\boldsymbol{K_1}=\operatorname{diag}\left(\sqrt{\alpha_1\omega_1},..., \sqrt{\alpha_n\omega_n}\right)$ and $\omega_r$ is a positive design parameter. 
	The equation \eqref{T2} is the differential Riccati equation for the Riccati
	filter which is used to apply ${\boldsymbol{\hat{H}}}^{-1}$ to the
	adaptive loop. Also the initial conditions $\boldsymbol{\Gamma}(0)$ must be chosen negative (positive) definite and symmetric if the cost function is convex (concave) map.
	\begin{figure}[t]
		\centering
		\includegraphics[scale = .5]{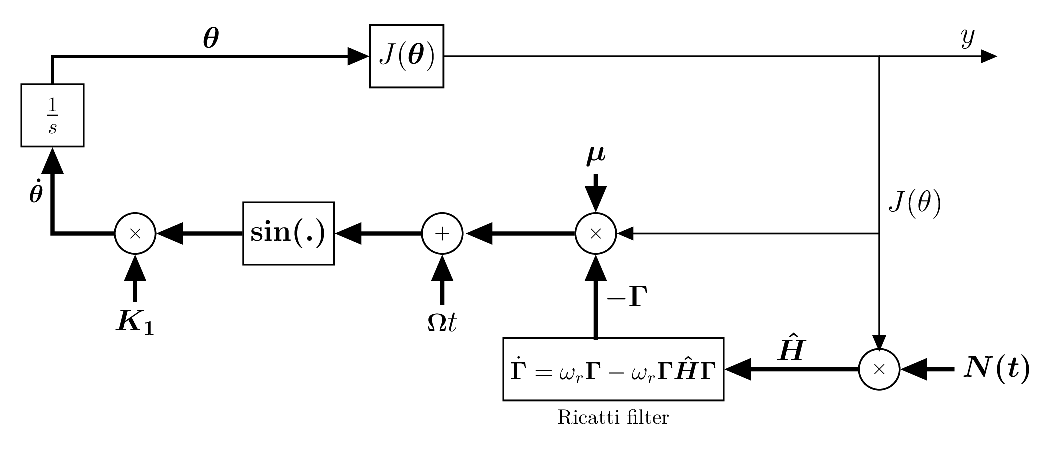}
		\caption{ Multivariable Newton-based Bounded ES scheme.}
		\label{fig2}
	\end{figure}
	\subsection{Hessian Estimation}\label{section:parthessian}
	The main idea of the Newton-based ES has been taken from \cite{ghaffari_multivariable_2012}, such that the suitable demodulation signals are designed to  produce, on average, the gradient and Hessian estimate of the cost function. To continue, the static cost function has been selected as \eqref{J}. It is assumed that the extremum point is a maximum, so $\boldsymbol{H}$ is negative definite. Let 
	\begin{subequations}
		\begin{align}
			\label{errorT}
			&\boldsymbol{\tilde{\theta}}=\boldsymbol{\theta}-\boldsymbol{\theta^*}+\boldsymbol{K_2}\cos(\boldsymbol{\scalebox{0.85}{${\Omega}$}} t+
			\boldsymbol{	{\Gamma}}\boldsymbol{\mu}J)\\
			\label{errorG}
			&\boldsymbol{\tilde{\Gamma}}=\boldsymbol{{\Gamma}}-\boldsymbol{H}^{-1}
		\end{align}
		\label{error}
	\end{subequations}
	be the estimation error of the extremum point $\boldsymbol{\theta^*}$ and the Hessian inverse. The cost function, $J$ in the new variable is
	\begin{align}
		\allowdisplaybreaks
		\nonumber
		J(\boldsymbol{\tilde{\theta}})=J(\boldsymbol{\theta^*})&+\frac{1}{2}\boldsymbol{\tilde{\theta}}^T \boldsymbol{H} \boldsymbol{\tilde{\theta}}-\frac{1}{2}\cos^T(\boldsymbol{\scalebox{0.85}{${\Omega}$}} t+\boldsymbol{\phi})\boldsymbol{K_2}^T \boldsymbol{H} \boldsymbol{\tilde{\theta}}\\
		\nonumber	&-\frac{1}{2} \boldsymbol{\tilde{\theta}}^T\boldsymbol{H}\boldsymbol{K_2}\cos(\boldsymbol{\scalebox{0.85}{${\Omega}$}} t+\boldsymbol{\phi})\\
		&+\frac{1}{2} \cos^T(\boldsymbol{\scalebox{0.85}{${\Omega}$}} t+\boldsymbol{\phi})\boldsymbol{K_2}^T\boldsymbol{H}\boldsymbol{K_2}\cos(\boldsymbol{\scalebox{0.85}{${\Omega}$}} t+\boldsymbol{\phi})
		\label{static-map-new2}
	\end{align}
	where $\boldsymbol{\phi}=(\boldsymbol{\tilde{\Gamma}}+\boldsymbol{H}^{-1})\boldsymbol{\mu}J$. 
	
	Now we aim to design the demodulation matrix $\boldsymbol{N}(t)$ in such a way that the average value of $\boldsymbol{N(t)}\times J(\boldsymbol{\tilde{\theta}})$ 
	provides the estimation of matrix $\boldsymbol{H}$.	First, multiplying the cost function \eqref{static-map-new2}  by  $\boldsymbol{K_2}^{\!\!\!-1} \cos(\boldsymbol{\scalebox{0.85}{${\Omega}$}} t+\boldsymbol{\phi})\cos^T(\boldsymbol{\scalebox{0.85}{${\Omega}$}} t+\boldsymbol{\phi})\boldsymbol{K_2}^{\!\!\!-1}$, we have
	\begin{align}
		\nonumber & \boldsymbol{\eta}=\boldsymbol{K_2}^{\!\!\!-1} \cos(\boldsymbol{\scalebox{0.85}{${\Omega}$}} t+\boldsymbol{\phi}) \times 	J(\boldsymbol{\tilde{\theta}})\times\cos^T(\boldsymbol{\scalebox{0.85}{${\Omega}$}} t+\boldsymbol{\phi})\boldsymbol{K_2}^{\!\!\!-1}\\
		\nonumber& =\hspace{0cm}\!\left(J(\boldsymbol{\theta^*})+\frac{1}{2}\boldsymbol{\tilde{\theta}}^T \boldsymbol{H} \boldsymbol{\tilde{\theta}}\right) \!\!\boldsymbol{K_2}^{\!\!\!-1} \cos(\boldsymbol{\scalebox{0.85}{${\Omega}$}} t+\boldsymbol{\phi})\cos^T\!(\boldsymbol{\scalebox{0.85}{${\Omega}$}} t+\boldsymbol{\phi})\boldsymbol{K_2}^{\!\!\!-1}\\
		\nonumber &\hspace{0.5cm}-\frac{1}{2}\boldsymbol{K_2}^{\!\!\!-1} \cos(\boldsymbol{\scalebox{0.85}{${\Omega}$}} t+\boldsymbol{\phi}) \!\!\left(\boldsymbol{\tilde{\theta}}^T\!\boldsymbol{H}\boldsymbol{K_2}\cos(\boldsymbol{\scalebox{0.85}{${\Omega}$}} t+\boldsymbol{\phi})\right)\\
		\nonumber&\hspace{3cm}\times	\cos^T(\boldsymbol{\scalebox{0.85}{${\Omega}$}} t+\boldsymbol{\phi})\boldsymbol{K_2}^{\!\!\!-1}\\
		&\hspace{0.5cm}+\frac{1}{2}\nonumber\left(\cos^T(\boldsymbol{\scalebox{0.85}{${\Omega}$}} t+\boldsymbol{\phi})\boldsymbol{K_2}^T\boldsymbol{H}\boldsymbol{K_2}\cos(\boldsymbol{\scalebox{0.85}{${\Omega}$}} t+\boldsymbol{\phi})\right)\\
		\nonumber&\hspace{3cm}\times \boldsymbol{K_2}^{\!\!\!-1} \cos(\boldsymbol{\scalebox{0.85}{${\Omega}$}} t+\boldsymbol{\phi})\cos^T(\boldsymbol{\scalebox{0.85}{${\Omega}$}} t+\boldsymbol{\phi})\boldsymbol{K_2}^{\!\!\!-1}\\
		& =\boldsymbol{\eta_1}+\boldsymbol{\eta_2}+\boldsymbol{\eta_3}
		\label{ave2}
	\end{align}
		To apply the averaging theory for obtaining $\boldsymbol{\eta}_{av}$, consider $\boldsymbol{\eta}$ as \eqref{ave2} in the time-scale $\tau=\omega t$ such that $\frac{1}{\omega}$ is a small positive parameter and $\boldsymbol{\tilde{\theta}}$ and $J(\boldsymbol{\theta^*})$ are recognized as constant. Consequently the average value of the first summand ($\boldsymbol{\eta_1}$) is
		\begin{align}
			\boldsymbol{\eta_{1,av}}=\frac{1}{2}\left(J(\boldsymbol{\theta^*})+\frac{1}{2}\boldsymbol{\tilde{\theta}_{av}}^T \boldsymbol{H} \boldsymbol{\tilde{\theta}_{av}}\right)\boldsymbol{K_2}^{\!\!\!-1}\boldsymbol{K_2}^{\!\!\!-1}
		\end{align}
		In a lengthy yet straightforward calculation, it can be demonstrated that the expression for the second summand  ($\boldsymbol{\eta}_2$) forms a matrix with elements that are expressed as various combinations of trigonometric functions, specifically involving cosine terms. These terms include: $
		\cos^3(\overline{{\omega}}_i\tau +{\phi}_i),\hspace{1mm} \cos^2(\overline{{\omega}}_i\tau +{\phi}_i) \cos(\overline{{\omega}}_j\tau +{\phi}_j)$ or $ \cos(\overline{{\omega}}_i\tau +{\phi}_i) \cos(\overline{{\omega}}_j\tau +{\phi}_j)\cos(\overline{{\omega}}_k\tau +{\phi}_k)$, for $i\neq j\neq k$. And the probing frequencies $\omega_i$  are chosen to satisfy \eqref{W_isump}, so the average of all elements in this matrix equals zero. As a result, the average of the second term effectively cancels out or vanishes.
		
		In order to express in more detail about the average of $\boldsymbol{\eta}_3$, it is rewritten in an element-wise manner, and  elements are 
		\begin{align}
			\allowdisplaybreaks
			\nonumber (\boldsymbol{\eta_3})_{ij}=&\frac{1}{2k_ik_j} \cos(\omega_i t + \phi_i) \cos(\omega_j t + \phi_j)\times \\
			&
			\sum_{m=1}^{n} \sum_{k=1}^{n} k_k k_m H_{km}	\cos(\omega_m t + \phi_m)\cos(\omega_k t + \phi_k) 
			\label{eta3ij}
		\end{align} where $k_i$ are the diagonal elements of matrix $\boldsymbol{K_2}$.	On the probing frequencies $\overline{{\omega}}_i$, we have
		\begin{align}
			\allowdisplaybreaks
			&\frac{1}{T}\int_{0}^{T}\cos^4(\overline{{\omega}}_i\tau+{\phi}_i)d\tau=\frac{3}{8}\\
			&\frac{1}{T}\int_{0}^{T}\cos^2 (\overline{{\omega}}_i\tau+{\phi}_i)\cos^2(\overline{{\omega}}_j\tau+{\phi}_j)d\tau=\frac{1}{4}, \hspace{0.3cm}  i\neq j
		\end{align}
		thus the average equation \eqref{eta3ij} becomes 
		\begin{align}
			(\boldsymbol{\eta_3})_{ij,av} =  \begin {cases}
			\frac{1}{2}\times \frac{3}{8}H_{ii}+\frac{1}{2k_i^2}\sum_{m=1,\neq i}^{n}\frac{k_m^2}{4}H_{mm} &   i=j \\ \frac{1}{8}(H_{ij}+H_{ji})= \frac{1}{4}H_{ij}&  i\neq j
		\end{cases}
	\end{align}
	Overall,  $\boldsymbol{\eta_{av}}$ is  given by
	
	\begin{align}
		\nonumber\boldsymbol{\eta_{av}}= &\frac{1}{2}\boldsymbol{K_2}^{\!\!\!-1}\!\!\left(\!J(\boldsymbol{\theta^*})+\frac{1}{2}\boldsymbol{\tilde{\theta}_{av}}^T \boldsymbol{H} \boldsymbol{\tilde{\theta}_{av}}+\frac{1}{4} tr(\boldsymbol{K_2} \boldsymbol{H }\boldsymbol{K_2})\!\right)\!\!\boldsymbol{K_2}^{\!\!\!-1}\\
		&\hspace{0.2 cm}+\frac{1}{4} \boldsymbol{H }-\frac{3}{16}\operatorname{diag}( {H }_{11},..., {H }_{nn})
		\label{ave3}
	\end{align}
	Note that the average of the cost function \eqref{static-map-new2} is
	\begin{align}
		J_{av} =J(\boldsymbol{\theta^*})+\frac{1}{2}\boldsymbol{\tilde{\theta}_{av}}^T \boldsymbol{H} \boldsymbol{\tilde{\theta}_{av}}+\frac{1}{4}  tr(\boldsymbol{K_2}^T \boldsymbol{H }\boldsymbol{K_2})
		\label{aveJ}
	\end{align}
	It can be seen that
	\begin{align}
		\nonumber4\boldsymbol{\eta_{av}}-2\boldsymbol{K_2}^{\!\!\!-1}J_{av}\boldsymbol{K_2}^{\!\!\!-1} &=\boldsymbol{H}-\frac{3}{4}\operatorname{diag}(H_{11},...,H_{nn})
		\\&\hspace{-1cm}=\left[
		\begin{array}{cccc}
			\frac{1}{4}H_{11} & H_{12} & ... & H_{1n} \\
			H_{12} & \frac{1}{4}H_{22} & ... & H_{2n} \\
			\colon & \colon &  &  \colon\\
			H_{1n} & ... &  & \frac{1}{4}H_{nn}
		\end{array}
		\right]
		\label{finalMM}
	\end{align}
	Therefore, from \eqref{ave2} and \eqref{finalMM}, an estimate of the Hessian can be obtained by multiplying the map $J$ with the matrix 
	$	4 \boldsymbol{K_2}^{\!\!\!-1} \cos(\boldsymbol{\scalebox{0.85}{${\Omega}$}} t+\boldsymbol{\phi})\cos^T(\boldsymbol{\scalebox{0.85}{${\Omega}$}} t+\boldsymbol{\phi})\boldsymbol{K_2}^{\!\!\!-1}-2\boldsymbol{K_2}^{\!\!\!-1}\boldsymbol{K_2}^{\!\!\!-1}$, except that its diagonal entries must be increased fourfold. Thus, the entries of the demodulation matrix $\boldsymbol{N}(t)$ are 
	\begin{align}
		N_{ij}(t) = \begin {cases}
		16k_i^2\left(\cos^2(\omega_i t+\phi_i)-\frac{1}{2}\right) &   i=j \\ 4k_ik_j\cos(\omega_i t+\phi_i)\cos(\omega_j t+\phi_j) &  i\neq j
	\end{cases}
\end{align}
\subsection{Stability Analysis}
This section establishes the stability of the proposed algorithm in the static-map setting. We first introduce the following assumption on the Hessian, which characterizes weak coupling among the coordinates.
\begin{assumption}\label{assumpH}
	The off-diagonal Hessian terms are assumed to be sufficiently small, i.e.,
	$
	\sum_{j\neq i}|H_{ij}|\leq \varepsilon_i |H_{ii}|,\quad 0<\varepsilon_i<1
	$.
\end{assumption}
The Hessian estimator developed in Section \ref{section:parthessian} does not require any restriction on the relative magnitude of the diagonal and off-diagonal Hessian entries. However, a condition on the Hessian coupling is introduced in the stability analysis of the proposed Newton-based bounded ES scheme.
We use the averaging method to analyze the local stability which is stated in the following theorem.
\begin{theorem}
	Consider the feedback system \eqref{System2} under Assumption \ref{assumpH}. 
	There exists  positive constants $\omega^*$ and $\alpha^*$ such that for all perturbation frequencies $\omega>\omega^*$ and \(0\!<\!\alpha_i\!<\!{\alpha}^*\), \(i\!=\!1,\ldots,n\),
	the solutions $(\boldsymbol{\theta},\boldsymbol{\Gamma})$ of \eqref{System2} are locally exponentially 
	attracted to a periodic solution lying in an $O\!\left(\frac{1}{\omega}+\underset{i}{\max}\sqrt{\frac{\alpha_i}{\omega_i}}\right)$-neighborhood of 
	$(\boldsymbol{\theta}^*,\boldsymbol{H}^{-1})$.
\end{theorem}
\begin{proof}
	At the beginning, by the time-derivative of \eqref{error}, the error dynamics is
	\begin{subequations}
		\begin{align}
			\nonumber&	\boldsymbol{\dot{\tilde{\theta}}}=\boldsymbol{K_1}\sin \left(\boldsymbol{\scalebox{0.85}{${\Omega}$}} t+\boldsymbol{\phi}\right)\\
			\nonumber&\hspace{0cm}-\boldsymbol{K_2}\left(\operatorname{diag}(\boldsymbol{\scalebox{0.85}{${\Omega}$}})+ \operatorname{diag}((\boldsymbol{\tilde{\Gamma}}+\boldsymbol{H}^{-1})\boldsymbol{\mu}) \frac{dJ}{dt}+\operatorname{diag}(\frac{d \boldsymbol{\tilde{\Gamma}}}{dt}\boldsymbol{\mu}) J\right)\\
			&\hspace{1.5cm}\times\sin \left(\boldsymbol{\scalebox{0.85}{${\Omega}$}} t+\boldsymbol{\phi}\right)
			\label{error_sytem1}
			\\
			&	\boldsymbol{\dot{{\tilde{\Gamma}}}}=\omega_r\!\big(\boldsymbol{\tilde\Gamma}+\boldsymbol{H}^{-1}\big)
			- \omega_r\!\big(\boldsymbol{\tilde\Gamma}+\boldsymbol{H}^{-1}\big)\,
			{\boldsymbol{N}}(t)\,
			J\!\big(\boldsymbol{\tilde\Gamma}+\boldsymbol{H}^{-1}\big)
		\end{align}
		\label{sys32}
	\end{subequations}
	As regards, $	J(\boldsymbol{\tilde{\theta}})$ in \eqref{static-map-new2}, we have
	\begin{align}
		\nonumber\frac{d	J}{dt}&=\boldsymbol{\dot{\theta}}^T\boldsymbol{H }(\boldsymbol{	\theta}-\boldsymbol{	\theta^*})\\
		&=\sin^T \left({\boldsymbol{\scalebox{0.85}{${\Omega}$}}}t+\boldsymbol{\phi}\right)\boldsymbol{ K_1H} 	\left(\boldsymbol{{\tilde{\theta}}}-\boldsymbol{K_2} \cos \left({\boldsymbol{\scalebox{0.85}{${\Omega}$}}}t+\boldsymbol{\phi}\right)\right)
		\label{DJ}
	\end{align}
	and by substituting \eqref{DJ} into \eqref{sys32}
	\begin{subequations}
		\begin{align}
			\nonumber&\frac{	d\boldsymbol{{\tilde{\theta}}}}{dt}=\boldsymbol{K_1}\sin \left(\boldsymbol{\scalebox{0.85}{${\Omega}$}} t+\boldsymbol{\phi}\right)-\boldsymbol{K_2}\operatorname{diag}(\boldsymbol{\scalebox{0.85}{${\Omega}$}})\sin \left(\boldsymbol{\scalebox{0.85}{${\Omega}$}} t+\boldsymbol{\phi}\right)
			\\
			\nonumber&\hspace{0.75cm}-\boldsymbol{K_2}\operatorname{diag}((\boldsymbol{\tilde{\Gamma}}+\boldsymbol{H}^{-1})\boldsymbol{\mu})\sin^T \left(\boldsymbol{\scalebox{0.85}{${\Omega}$}} t+\boldsymbol{\phi}\right)\boldsymbol{K_1H} \boldsymbol{{\tilde{\theta}}}\\
				\nonumber	&\hspace{4cm}\times\sin \left(\boldsymbol{\scalebox{0.85}{${\Omega}$}} t+\boldsymbol{\phi}\right)\\
			\nonumber	&\hspace{0.75cm}+\boldsymbol{K_2}\operatorname{diag}((\boldsymbol{\tilde{\Gamma}}+\boldsymbol{H}^{-1})\boldsymbol{\mu}) \sin^T \left(\boldsymbol{\scalebox{0.85}{${\Omega}$}} t+\boldsymbol{\phi}\right)\boldsymbol{K_1}\boldsymbol{H K_2} \\
			\nonumber	&\hspace{4cm}\times\cos \left({\boldsymbol{\scalebox{0.85}{${\Omega}$}}}t+\boldsymbol{\phi}\right) \sin \left({\boldsymbol{\scalebox{0.85}{${\Omega}$}}}t+\boldsymbol{\phi}\right)
			\\
			&\hspace{0.75cm}-\boldsymbol{K_2}\operatorname{diag}(\frac{d \boldsymbol{\tilde{\Gamma}}}{dt} \boldsymbol{\mu})J\sin \left(\boldsymbol{\scalebox{0.85}{${\Omega}$}} t+\boldsymbol{\phi}\right)
			\label{error_sytem1111}\\
			&	\frac{d\boldsymbol{{{\tilde{\Gamma}}}}}{dt}=\omega_r\!\big(\boldsymbol{\tilde\Gamma}+\boldsymbol{H}^{-1}\big)
			- \omega_r\!\big(\boldsymbol{\tilde\Gamma}+\boldsymbol{H}^{-1}\big)\,
			{\boldsymbol{N}}(t)\,
			J\!\big(\boldsymbol{\tilde\Gamma}+\boldsymbol{H}^{-1}\big)
		\end{align}
	\end{subequations}
			They are chosen $\boldsymbol{\scalebox{0.85}{${\Omega}$}}=\omega \overline{\boldsymbol{\scalebox{0.85}{${\Omega}$}}}$ where $\overline{\boldsymbol{\scalebox{0.85}{${\Omega}$}}}=\left(\overline{\omega}_1,..., \overline{\omega}_n\right)^T$, and $\boldsymbol{K_2}=\frac{1}{\omega}\operatorname{diag}(\overline{\boldsymbol{\scalebox{0.85}{${\Omega}$}}})^{-1}\boldsymbol{K_1}$.
			Then, in the new time-scale $\tau=\omega t$, we have 
			\begin{subequations}
				\allowdisplaybreaks
				\begin{align}
					\nonumber&\frac{	d\boldsymbol{{\tilde{\theta}}}}{d\tau}=\frac{1}{\omega}\Big(\boldsymbol{K_1}\sin \left(\overline{\boldsymbol{\scalebox{0.85}{${\Omega}$}}} \tau+\boldsymbol{\phi}\right)-\boldsymbol{K_1}	\sin
					\left(\overline{\boldsymbol{\scalebox{0.85}{${\Omega}$}}} \tau+\boldsymbol{\phi}\right)\Big)\\
					\nonumber&\hspace{0.6cm}-\frac{1}{\omega}\Big(\overline{\boldsymbol{K_1}}\operatorname{diag}(\overline{\boldsymbol{\scalebox{0.85}{${\Omega}$}}})^{-1}\operatorname{diag}(\boldsymbol{(\boldsymbol{\tilde{\Gamma}}+\boldsymbol{H}^{-1})\mu})\sin^T \left(\overline{\boldsymbol{\scalebox{0.85}{${\Omega}$}}} \tau+\boldsymbol{\phi}\right)\\
					\nonumber&\hspace{3.7cm}\times	\boldsymbol{H }\overline{\boldsymbol{K_1}} \boldsymbol{{\tilde{\theta}}}\sin \left(\overline{\boldsymbol{\scalebox{0.85}{${\Omega}$}}} \tau+\boldsymbol{\phi}\right)\Big)\\
					\nonumber&\hspace{0.6cm}+\frac{1}{\omega^3}\Big(\boldsymbol{K_1 }\operatorname{diag}(\boldsymbol{(\boldsymbol{\tilde{\Gamma}}+\boldsymbol{H}^{-1})\mu})\sin^T \left(\overline{\boldsymbol{\scalebox{0.85}{${\Omega}$}}} \tau+\boldsymbol{\phi}\right)\overline{\boldsymbol{K_1}}\boldsymbol{H}\overline{\boldsymbol{K_1}}\\
					\nonumber &\hspace{2cm}\times\operatorname{diag}(\overline{\boldsymbol{\scalebox{0.85}{${\Omega}$}}})^{-1}\cos \left(\overline{\boldsymbol{\scalebox{0.85}{${\Omega}$}}}\tau+\boldsymbol{\phi}\right)\sin \left(\overline{\boldsymbol{\scalebox{0.85}{${\Omega}$}}}\tau+\boldsymbol{\phi}\right)\Big)
					\\
					&\hspace{0.6cm}-\frac{1}{\omega^3}\Big(\boldsymbol{K_1}\operatorname{diag}(\overline{\boldsymbol{\scalebox{0.85}{${\Omega}$}}})^{-1}\operatorname{diag}(\boldsymbol{G_\gamma}\boldsymbol{\mu})J\sin \left(\overline{\boldsymbol{\scalebox{0.85}{${\Omega}$}} }\tau+\boldsymbol{\phi}\right)\Big)\\
					\label{error_sytem13}
					\nonumber&\frac{d\boldsymbol{{{\tilde{\Gamma}}}}}{d\tau}=\frac{1}{\omega}\left( \omega_r\!\big(\boldsymbol{\tilde\Gamma}+\boldsymbol{H}^{-1}\big)
					\!- \omega_r\!\big(\boldsymbol{\tilde\Gamma}+\boldsymbol{H}^{-1}\big)\,
					\overline{\boldsymbol{N}}(\tau)\,
					\!J\!\big(\boldsymbol{\tilde\Gamma}+\boldsymbol{H}^{-1}\big)  \!\!  \right)\\
					&\hspace{0.6cm}	=\frac{1}{\omega}\boldsymbol{G_\gamma}
				\end{align}
				\label{sys333}
			\end{subequations}
			where $\overline{\boldsymbol{K_1}}=\operatorname{diag}\left(\sqrt{\alpha_1\overline{\omega}}_1,..., \sqrt{\alpha_n\overline{\omega}}_n\right)$ and $\overline{\boldsymbol{{N}}}(\tau)=\boldsymbol{N}(\tau/\omega)$. If $\epsilon=\frac{1}{\omega}$, 
			for the sufficiently large $\omega$, 
			\begin{subequations}
				\begin{align}
					\allowdisplaybreaks
					&\frac{	d\boldsymbol{{\tilde{\theta}}}}{d\tau}=\epsilon \boldsymbol{f_\theta}(\boldsymbol{\tilde{\theta}},\boldsymbol{\tilde{\Gamma}},\epsilon,\tau)
					\label{error_sytem133}
					\\
					&\frac{d\boldsymbol{{{\tilde{\Gamma}}}}}{d\tau}=\epsilon \boldsymbol{f_\gamma}(\boldsymbol{\tilde{\theta}},\boldsymbol{\tilde{\Gamma}},\epsilon,\tau)
				\end{align}
				\label{sys3333}
			\end{subequations}
			in which
			\begin{subequations}
				\begin{align}
					\nonumber&\boldsymbol{f_\theta}(\boldsymbol{\tilde{\theta}},\boldsymbol{\tilde{\Gamma}},\epsilon,\tau)=	-\operatorname{\operatorname{diag}}(\overline{\boldsymbol{\scalebox{0.85}{${\Omega}$}}}^{-1})
					\overline{\boldsymbol{K_1}}\operatorname{diag}(\boldsymbol{(\boldsymbol{\tilde{\Gamma}}+\boldsymbol{H}^{-1})\mu})\\
					&\hspace{0.5cm}\times\sin^T\!\! \left(\overline{\boldsymbol{\scalebox{0.85}{${\Omega}$}}} \tau+\boldsymbol{\phi}\right)	\overline{\boldsymbol{K_1}}\boldsymbol{H} \boldsymbol{{\tilde{\theta}}}\sin \left(\overline{\boldsymbol{\scalebox{0.85}{${\Omega}$}}} \tau+\boldsymbol{\phi}\right)+\epsilon^2 \boldsymbol{G}_\theta(\boldsymbol{\tilde{\theta}},\boldsymbol{\tilde{\Gamma}},\epsilon,\tau)
					\\[4pt]&
					\nonumber	\boldsymbol{f_\gamma}(\boldsymbol{\tilde\theta},\boldsymbol{\tilde\Gamma},\epsilon,\tau)
					= \omega_r\!\big(\boldsymbol{\tilde\Gamma}+\boldsymbol{H}^{-1}\big)
					- \omega_r\!\big(\boldsymbol{\tilde\Gamma}+\boldsymbol{H}^{-1}\big)\,
					\overline{\boldsymbol{N}}(\tau)\,\\
			&	\hspace{5.35cm}\times
					J\!\big(\boldsymbol{\tilde\Gamma}+\boldsymbol{H}^{-1}\big)
				\end{align}
			\end{subequations}
				where $\boldsymbol{G_\theta}$ aggregates the higher-order terms proportional to $\epsilon^3$
				in the $\boldsymbol{{\tilde{\theta}}}$-dynamics. The system \eqref{sys3333} is $\Pi$-periodic in $\tau$ and in the standard form of the averaging theorem. The averaging period is also given
						\begin{align}
							\Pi=2\pi\times LCM\{\frac{1}{\overline{\omega}_1},...,\frac{1}{\overline{\omega}_n}\},\hspace{0.75cm} i\in\{1,2,...,n\}
						\end{align}
						The average of functions $\boldsymbol{f_{\theta,av}}(\boldsymbol{\tilde{\theta}},\boldsymbol{\tilde{\Gamma}})$ and  $\boldsymbol{f_{\gamma,av}}(\boldsymbol{\tilde{\theta}},\boldsymbol{\tilde{\Gamma}})$ equal
						\begin{subequations}
							\begin{align}
								&	\nonumber\boldsymbol{f_{\theta,av}}(\boldsymbol{\tilde{\theta}},\boldsymbol{\tilde{\Gamma}})=\frac{1}{\Pi}\int_{0}^{\Pi}\boldsymbol{f_\theta}(\boldsymbol{\tilde{\theta}},\boldsymbol{\tilde{\Gamma}},0,\tau)d\tau\\
								&\hspace{1cm}	=\!\!\frac{-1}{2}\operatorname{diag}(\overline{\boldsymbol{\scalebox{0.85}{${\Omega}$}}}^{-1})\overline{\boldsymbol{K}}_{\boldsymbol{1}}\operatorname{diag}(\boldsymbol{(\boldsymbol{\tilde{\Gamma}}_{av}+\boldsymbol{H}^{-1})\mu})	\overline{\boldsymbol{K}}_{\boldsymbol{1}}\boldsymbol{H} \boldsymbol{{\tilde{\theta}}_{av}}\\
								&\nonumber\boldsymbol{f_{\gamma,av}}(\boldsymbol{\tilde{\theta}},\boldsymbol{\tilde{\Gamma}})=\frac{1}{\Pi}\int_{0}^{\Pi}\boldsymbol{f_\gamma}(\boldsymbol{\tilde{\theta}},\boldsymbol{\tilde{\Gamma}},0,\tau)d\tau\\
								&\hspace{1cm}	=-\omega_r\boldsymbol{\tilde\Gamma_{av}}\boldsymbol{H}\boldsymbol{\tilde\Gamma_{av}}-\omega_r\boldsymbol{\tilde\Gamma_{av}}
								\label{AvgGamma}
							\end{align}
							\label{avg:teta:gama}
						\end{subequations}
						Under Assumption \ref{assumpH}, the corresponding averaged system is given by
						\begin{subequations}
							\begin{align}
								&
								\boldsymbol{	\dot{\tilde{\theta}}_{av}}=	-\frac{\epsilon}{2}\boldsymbol{K^\prime} \boldsymbol{{\tilde{\theta}}}_{av}\\
								&		\boldsymbol{	\dot{\tilde{\Gamma}}_{av}}=-\epsilon\omega_r\boldsymbol{\tilde\Gamma_{av}}\boldsymbol{H}\boldsymbol{\tilde\Gamma_{av}}-\epsilon\omega_r\boldsymbol{\tilde\Gamma_{av}}
							\end{align}
							\label{avgsystemzm1}
						\end{subequations}
						where $\boldsymbol{K^\prime}
						=\operatorname{diag}
					(	\alpha_i
						\left(
						H_{ii}\tilde{\Gamma}_{i,\mathrm{av}}+\mu_i
						\right))$.
						The right-hand sides of \eqref{avgsystemzm1} are equated to zero to determine equilibrium conditions. Taking into account  initial conditions $\boldsymbol{\Gamma}(0)$ for convex map and $\overline{\omega}_r> 0$, the average equilibriums give
							\begin{subequations}
								\begin{align}
									&\boldsymbol{{\tilde{\theta}_{av}}}=0\\
									&
									\!\big(\boldsymbol{\tilde\Gamma_{av}}+\boldsymbol{H}^{-1}\big)\boldsymbol{H}=\boldsymbol{I}
								\end{align}
							\end{subequations}
							then the equilibrium point is $( \boldsymbol{{\tilde{\theta}_{av}}},\boldsymbol{\tilde\Gamma_{av}}) = (0,0)$, and
							the Jacobian matrix of the average system \eqref{avgsystemzm1} is
								\begin{align}
									\boldsymbol{A} ={\epsilon}\begin{bmatrix}
										-\boldsymbol{\rho} & \boldsymbol{0_{n\times n^2}} \\[1em]
										\boldsymbol{0_{n^2\times n}} &- \omega_r\boldsymbol{I_{n^2}}
									\end{bmatrix}
								\end{align}
								where $\boldsymbol{\rho}=\frac{1}{2}\operatorname{diag}(\alpha_1\mu_1,..., \alpha_n\mu_n)\in \mathbb{R}^{n\times n}$, so $\boldsymbol{A}$ is Hurwitz. The
								equilibrium of the average system \eqref{avgsystemzm1} is locally exponentially stable. Therefore, the application of the averaging theorem \cite{khalil2002nonlinear}[ Th. 10. 4] proves the existence of a unique exponentially
								stable periodic solution   $( \boldsymbol{{\tilde{\theta}}}, \boldsymbol{\tilde{\Gamma}})$ of the original system, with period $\Pi$, which remains within $O\!\left( \frac{1}{\omega} \right)$-neighborhood of the equilibrium of the averaged system, for sufficiently large $\omega$.
								Thus, regarding to \eqref{error}, the solutions $(\boldsymbol{\theta},\boldsymbol{\Gamma})$ are locally exponentially attracted to a periodic solution within
								$O\!\left(\frac{1}{\omega}+\underset{i}{\max}\sqrt{\frac{\alpha_i}{\omega_i}}\right)$-neighborhood 
								of $(\boldsymbol{\theta}^*,\boldsymbol{H}^{-1})$.
							\end{proof}
							\section{Simulation}\label{Sim}
							To illustrate the results and highlight the difference between the bounded ES  scheme \cite{scheinker_extremum_2014} and the proposed Newton-based Bounded ES  scheme, the following multivariable cost function  is considered
							\begin{align}
								y=J(\boldsymbol{\theta})=J(\boldsymbol{\theta^*})+\frac{1}{2}(\boldsymbol{\theta}-\boldsymbol{\theta^*})^T \begin{bmatrix} 3&1 \\1& 4.5\end{bmatrix}(\boldsymbol{\theta}-\boldsymbol{\theta^*})
							\end{align}
							where $\boldsymbol{	\theta^*} = \begin{bmatrix} 2 ,4 \end{bmatrix}^T$. 
							For a fair comparison, all parameters are chosen the same except the gain matrix $\boldsymbol{{K_1}}$. We should select $\boldsymbol{{K_{1n}}}$ and $\boldsymbol{{K_{1g}}}$ such that $\boldsymbol{{K_{1g}}}=- \boldsymbol{{K_{1n}}}\boldsymbol{\Gamma}(0)$. Here, $\boldsymbol{{K_{1g}}}$ and $\boldsymbol{{K_{1n}}}$ represent the gain matrices for the gradient-based and Newton-based approaches, respectively.
							\begin{figure}[bt]
								\centering
								\includegraphics[scale=0.31]{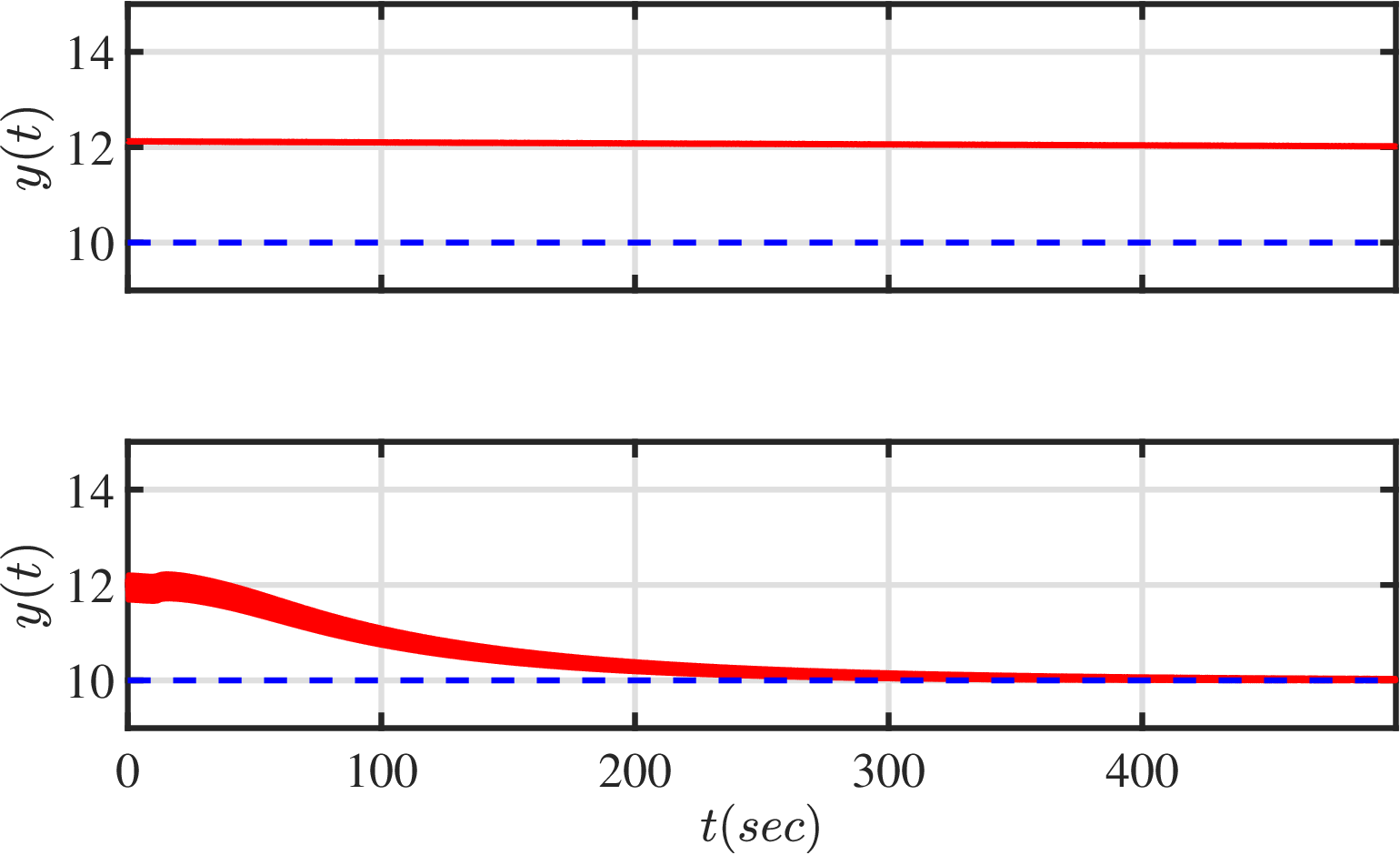}
								\caption{\fontsize{8}{10}\selectfont Two estimates of the maximum versus time: (above) Gradient-based and (below) Newton-based.}
								\label{figys}
							\end{figure}
							\begin{figure}[bt]
								\centering
								\includegraphics[scale=0.31]{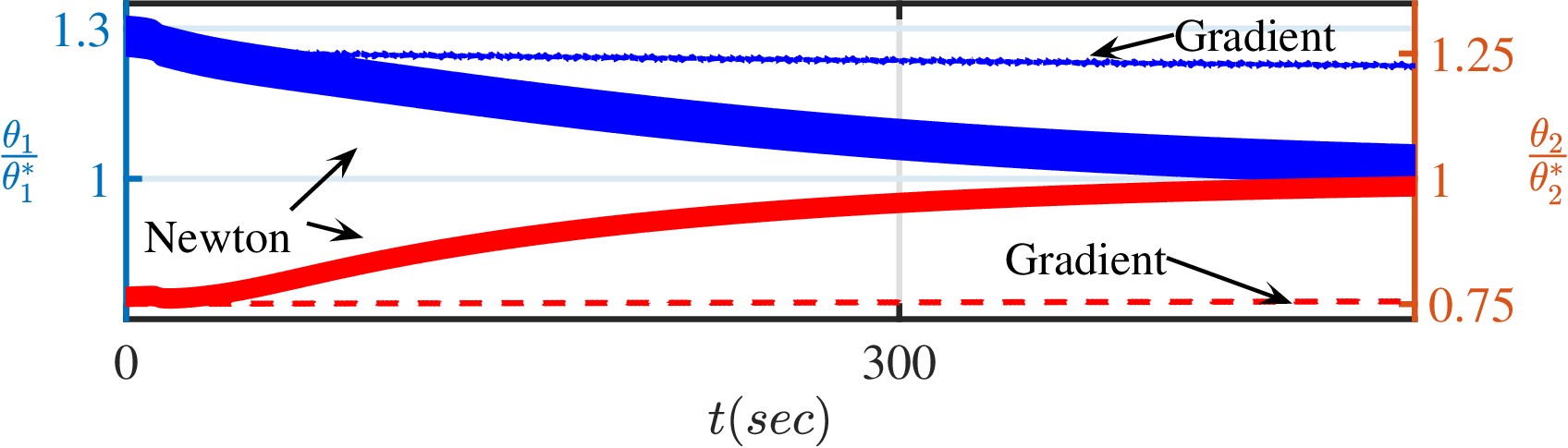}
								\caption{\fontsize{8}{10}\selectfont Parameter estimates. Gradient-based (blue)and Newton-based  (red) time responses.}
								\label{figFaws}
							\end{figure}
							\begin{figure}[bt]
								\centering
								\includegraphics[scale=0.3]{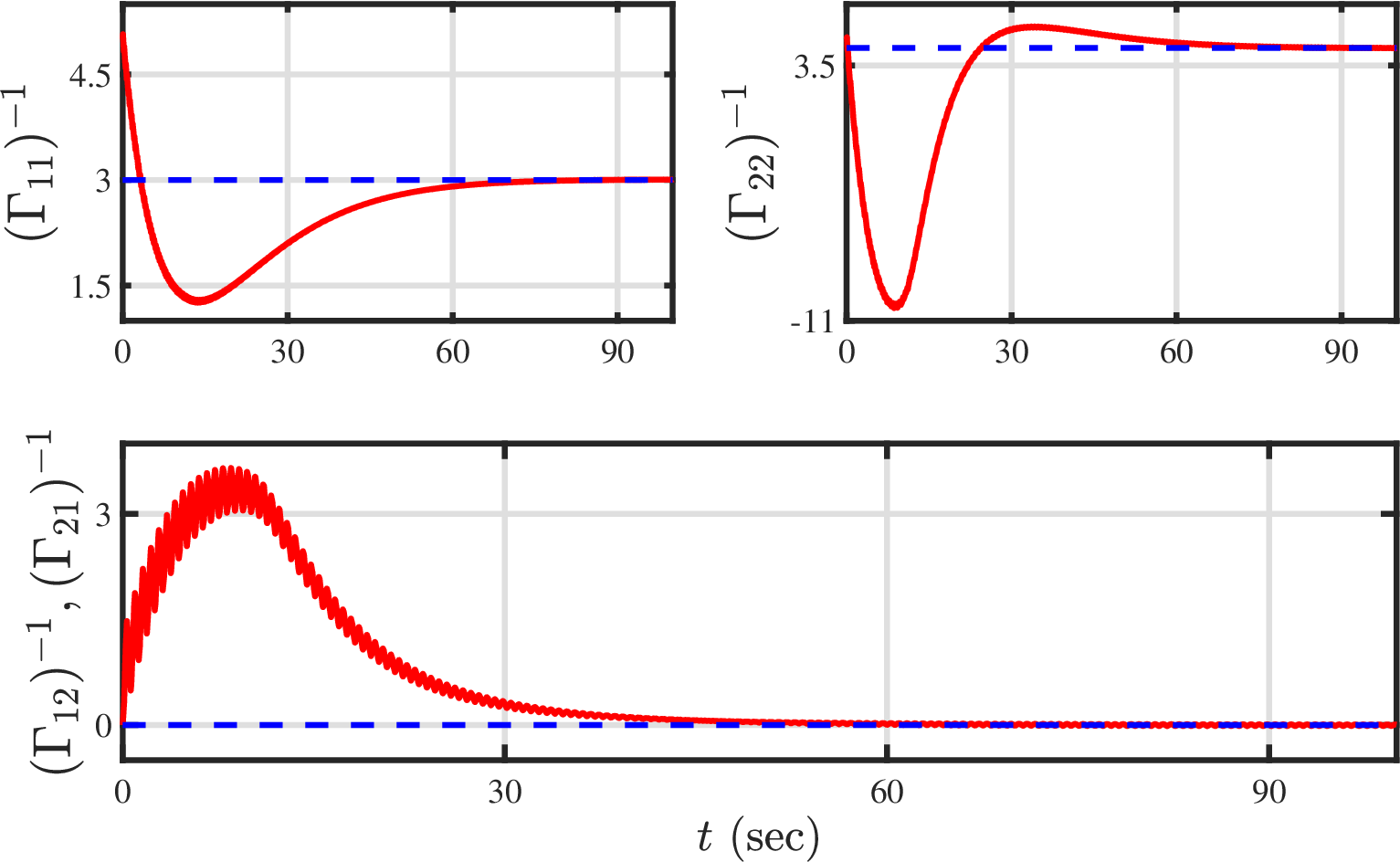}
								\caption{\fontsize{8}{10}\selectfont Time evolution of the Hessian matrix estimator $\boldsymbol{\Gamma^{-1}}$. The true value of $\boldsymbol{H}$ is reached in 60 seconds, after which the Newton-based method follows a direct path.}
								\label{figGamma}
							\end{figure}
							We perform our tests with the following parameters: 
							$\omega=10, \overline{\boldsymbol{\scalebox{0.85}{${\Omega}$}}}=\begin{bmatrix}
								2, 3
							\end{bmatrix}^T, \boldsymbol{a}= \begin{bmatrix} 0.1 , 0.1 \end{bmatrix}^T, \boldsymbol{\mu}= \begin{bmatrix} 0.2 , 0.2 \end{bmatrix}^T, \boldsymbol{	\Gamma}^{-1}(0)= \begin{bmatrix}
								0.1,-0.001;-0.001,0.02
							\end{bmatrix}, \boldsymbol{\theta}(0)= \begin{bmatrix} 2.5 , 3 \end{bmatrix}^T$.
							
							Fig. \ref{figys} illustrates the performance and convergence characteristics of two methods in achieving the maximum value. Evolution of the parameters is depicted in Fig. \ref{figFaws}. As illustrated in Fig. \ref{figys} and Fig. \ref{figFaws}, the Newton-based approach attains the maximum value within 300 seconds, whereas the gradient-based method lags considerably behind, emphasizing the faster convergence rate of the Newton-based technique. Furthermore, Fig. \ref{figGamma} depicts the convergence of the Hessian matrix to its true value, underscoring the reliability and efficiency of the Newton-based technique.
							\section{CONCLUSION}\label{CoN}
							This paper proposed a multivariable Newton-based extremum-seeking scheme with bounded update rates. The method combines bounded ES with an online estimate of the inverse Hessian to compensate for the effect of unknown local curvature. A suitable demodulation matrix was developed for Hessian estimation, together with a Riccati-based adaptation law for its inverse.
							
							Local exponential stability was established for both the conventional bounded ES and the proposed Newton-based scheme using averaging analysis. For the Newton-based scheme, the stability result requires a weak Hessian-coupling condition as a sufficient assumption for the stability analysis. This condition is introduced as a sufficient technical assumption for the stability analysis and is not required for the construction of the Hessian estimator itself. Under sufficiently large perturbation frequencies and small perturbation amplitudes, the system converges exponentially to a neighborhood of the extremum and the inverse Hessian.
							
							Simulation results demonstrate the effectiveness of the Hessian estimation and Newton-based adaptation, showing faster convergence than gradient-based bounded ES while maintaining bounded parameter-update rates. Future work will focus on relaxing the Hessian-coupling condition and extending the method to more strongly coupled objective maps and higher-order derivative extremum seeking.
							\bibliographystyle{ieeetran}
							\bibliography{dispaperRef}
						\end{document}